\documentclass[journal,final]{IEEEtran}
\usepackage{epsfig,rotating,setspace,latexsym,amsmath,epsf,amssymb,bm}
\usepackage{color}

\usepackage{amsmath}
\usepackage{rotating,latexsym,amsmath,amssymb,bm}
\usepackage{cite}
\usepackage{amsthm}
\usepackage{subfigure}
\usepackage{graphicx}

\def\nb0{{\mathbf{0}}}
\def\nb1{{\mathbf{1}}}

\def\ncalA{{\mathcal{A}}}
\def\ncalB{{\mathcal{B}}}
\def\ncalC{{\mathcal{C}}}

\def\ncalH{{\mathcal{H}}}

\def\ncalO{{\mathcal{O}}}

\def\ncalS{{\mathcal{S}}}

\def\ncalW{{\mathcal{W}}}

\def\nbbN{{\mathbb{N}}}

\def\nbbR{{\mathbb{R}}}

\newtheorem{theorem}{Theorem}

\usepackage{bbm}
\usepackage{url}

\usepackage{wrapfig}
\usepackage{dsfont}
\usepackage{wasysym}
\usepackage{hyperref}
\usepackage{balance}
\usepackage[ruled,vlined]{algorithm2e}

\theoremstyle{plain}

\newtheorem{prop}{Proposition}

\theoremstyle{definition}
\newtheorem{Def}{Definition}
\newtheorem{assum}{Assumption}

\usepackage{thmtools}
\declaretheoremstyle[
  spaceabove=\topsep, spacebelow=\topsep,
  headfont=\normalfont\bfseries,
  notefont=\mdseries, notebraces={(}{)},
  bodyfont=\normalfont,
  postheadspace=1em,
  qed=\qedsymbol
]{mythmstyle}

\declaretheoremstyle[
  spaceabove=\topsep, spacebelow=\topsep,
  headfont=\normalfont\bfseries,
  notefont=\mdseries, notebraces={(}{)},
  bodyfont=\normalfont,
  postheadspace=1em,
  qed=$\diamond$
]{mythmstyle}
\declaretheorem[style=mythmstyle]{remark}

\DeclareMathOperator{\psaa}{\pi_{\text{SAA}}}
\DeclareMathOperator{\pts}{\pi_{\text{TS}}}

\begin{document}

\allowdisplaybreaks

\sloppy

\title{When is Cognitive Radar Beneficial?}

\author{C.E. Thornton and R.M. Buehrer
\thanks{The authors are with Wireless $@$ Virginia Tech, Bradley Department of ECE, Blacksburg, VA, USA, 24061. Correspondence: $thorntonc@vt.edu$.}}

\maketitle
\thispagestyle{plain}
\pagestyle{plain}
\vspace{-1cm}
\begin{abstract}
When should an online reinforcement learning-based frequency agile cognitive radar be expected to outperform a rule-based adaptive waveform selection strategy? We seek insight regarding this question by examining a dynamic spectrum access scenario, in which the radar wishes to transmit in the widest unoccupied bandwidth during each pulse repetition interval. Online learning is compared to a fixed rule-based sense-and-avoid strategy. We show that given a simple Markov channel model, the problem can be examined analytically for simple cases via \emph{stochastic dominance}. Additionally, we show that for more realistic channel assumptions, learning-based approaches demonstrate greater ability to generalize. However, for short time-horizon problems that are well-specified, we find that machine learning approaches may perform poorly due to the inherent limitation of convergence time. We draw conclusions as to when learning-based approaches are expected to be beneficial and provide guidelines for future study. 
\end{abstract}

\begin{IEEEkeywords}
cognitive radar, machine learning, radar signal processing, radar-communications coexistence
\end{IEEEkeywords}

\section{Introduction}
According to Dr. Smith \emph{et al.} \emph{``In the simplest of terms, cognitive radar could be thought of as a system of hardware and software in which the transmit and receive parameters (such as power, pulse length, pulse repetition frequency (PRF), modulation, frequency and polarization) are selected, in real-time, in response to the observed scene to optimize performance for a given application."} \cite[Introduction]{smith2016experiments}. Accepting this broad definition, one may then wonder under which physical conditions such a system might provide drastic performance improvements as compared to traditional radar. It has been demonstrated over a long period of time, in many works, that adaptive waveform selection and receive processing can provide large performance benefits compared to using a fixed traditional waveform, such as linear frequency modulation (LFM) \cite{sira2009waveform,horne2020fast,mitchell2018cost}. Further, adaptive processing has been shown to provide benefits when the scene is dynamically varying \cite{metcalf2015machine}.

Clearly, the specific mechanism used to guide the adaptive parameter selection is extremely important for cognitive radar \cite{mitchell2018cost,bell2015cognitive,smith2013coupled}. For example, one may apply a pre-conceived fixed decision rule, which allows the radar to adapt to its scene dynamically, in some sense meeting Smith's definition above. If the fixed decision rule is chosen appropriately, perhaps by using incorporating \emph{a priori} information about the radar's scene, this approach can be expected to perform quite well. This fixed decision approach differs from \emph{learning-based} systems, which aim to make fewer \emph{a priori} assumptions about the radar scene's physical properties and arrive at a decision rule through repeated experience.

Thus, a more difficult, yet relevant question to answer is whether an online \emph{utility maximizing} cognitive radar can outperform a fixed adaptive strategy, and if so, under what conditions. This is the question we pursue in the present work, and is relevant given the widespread expectation that cognitive radar systems engage in learning over time \cite[Section I]{haykin2012cognitive}. Several recent contributions have examined learning-based approaches in the context of cognitive radar \cite{smith2020neural,charlish2020implementing,thornton2021constrained,thornton2022universal,thornton2022online}. However, these works generally consider a fixed objective function, and examine the radar's performance under scenarios in which the learning approach is definitively beneficial. Additionally, the reliability of learning algorithms is an important issue which has been heavily scrutinized in recent years \cite{martin2020stochastically,chan2019measuring}. In applications such as radar, outage events are meant to be avoided at all costs, as missed detections and lost tracks are very costly. Thus, it is important to understand when learning-based approaches perform worse than simple alternatives, which may have more predictable points of failure.

Herein, we develop initial insight towards the general problem of determining the value of learning in cognitive radar. We proceed by considering a dynamic spectrum sharing scenario in which an adaptive radar system must coexist with a non-cooperative communications system, where the interfering system transmits according to a fixed Markov transition kernel. This scenario is inherently time-varying, and thus adaptive radar definitively provides better performance than a fixed-channel approach \cite{martone2021closing}. As a baseline adaptive strategy, we consider a sense-and-avoid policy, which attempts to avoid interference based on the radar's most recent observation of the interference channel. However, the fixed nature of the sense-and-avoid approach is limited in that it does not adapt its decision rule based on the particular dynamics of the interference. We see that a learning-based approach, which attempts to balance a trade-off between missed opportunities and collisions provides substantial benefits over SAA, especially when the interference behavior is time-varying and close to deterministic.

\emph{Contributions and outline:} In Section \ref{se:syst} we develop machinery to examine the performance of decision rules in Markov interference channels. In Section \ref{se:analysis} we show that any fixed decision rule in a Markov interference channel can be decomposed, and performance explicitly analyzed over a long time horizon. We use notions of stochastic dominance to show that reasonably-constructed ML approaches weakly dominate sense-and-avoid asymptotically. Section \ref{se:numer} provides numerical justification in the finite-horizon regime, while Section \ref{se:concl} provides concluding remarks.

\emph{Implications for system design:} Stochastic control approaches are generally less predictable than applying an adaptive strategy based on a fixed decision rule. However, reasonably constructed learning algorithms will generally perform at least as well as rule-based strategies in the asymptotic regime. System designers should consider observation probability, number of actions, and time horizon when deciding whether a ML approach is beneficial relative to a rule-based scheme. Since there exist conditions where one approach is more effective, we postulate that robust systems may choose to incorporate a mixture of learning-based and fixed decision rules\footnote{Such fixed decision rules could be learned ``offline" by a learning algorithm trained in a general setting.} in a hierarchical `metacognitive` framework \cite{mishra2020toward,martone2020metacognition}. 

\section{System Model}
\label{se:syst}
We consider a finite-horizon, discrete time, dynamic spectrum sharing scenario, where time is indexed by $t = 1,2,...,n$. A tracking radar must share a fixed channel with one or more communications systems. Each time step corresponds to a radar pulse repetition interval (PRI). It is assumed that the radar must transmit every PRI. During each PRI, the radar has access to an observation of the interference in a fixed channel of bandwidth $B$. The shared channel is divided into $N_{SB}$ equally-spaced sub-channels. The radar may transmit in any \emph{contiguous} grouping of sub-channels.

Let the interference state be represented by a $d$-element binary vector expressed by $s_{t} \in \ncalS \subseteq \{0,1\}^{d}$. The frequency content of the radar's waveform is similarly given by the binary vector $w_{t} \in \ncalW \{0,1\}^{d}$. Since the radar is limited to transmitting in contiguous sub-channels, the total number of allowable waveforms is $|\ncalW| = d(d+1)/2$. The radar's observation is an estimate of the current interference state $o_{t} \in \{0,1\}^{d}$. The interference channel behavior is completely specified by initial state $s_{0}$ and transition probability matrix $P$, with elements $p_{ij} = \mathbb{P}(s_{t} = j |s_{t-1} = i) \in [0,1]$. We note that the transition probability matrix, and therefore interference channel behavior, are unknown to the radar \emph{a priori}. Based on the underlying interference state and the choice of radar waveform, the radar receives a real-valued loss at each PRI, denoted by $\ell_{t}$. The exact specification of the loss mapping will be developed below.

The baseline rule-driven adaptive approach considered in this paper is the sense-and-avoid (SAA) policy, a form of which is implemented using low-cost hardware in \cite{Kirk2019,Kirk2020} and is defined here as follows:
\begin{Def}[Sense-and-Avoid Policy]
	Let $\pi_{\text{SAA}}: \ncalO \mapsto \ncalW$ be a fixed decision function which selects the widest contiguous bandwidth available assuming $o_{t} = o_{t-1}$. The decision function is implemented algorithmically as follows. 
	\begin{enumerate}
		\item The position of the longest run of zeros in vector $o_{t-1}$ is found.
		\item Waveform $w_{k}$ is selected such that the radar transmits in the widest available bandwidth.
		\item If there are multiple vacancies of the same length, ties are broken randomly.
	\end{enumerate}	 
\end{Def}

A variety of learning-based algorithms can be employed for online waveform selection. For comparison's sake, we focus predominantly on the Thompson Sampling (TS) approach, which is useful due to the simplicity of implementation, impressive empirical performance, and theoretical performance guarantees. The TS policy involves picking the waveform which has the highest posterior probability of yielding the lowest loss, and was introduced in \cite{thornton2021constrained}. 

\begin{Def}[Thompson Sampling Policy]
	The Thompson Sampling (TS) policy selects the waveform $w^{*}$ which has the highest posterior probability of yielding the lowest loss. More precisely, let $\ncalH$ be the set of waveforms, losses, and observations seen by the radar until PRI $t$, and let $\theta \in \Theta$ be a parameter employed by the radar to account for stochastic effects of the channel. The TS policy then proceeds as follows,
	\begin{multline}
	\pi_{\text{TS}} = \\ \int_{\Theta} \mathbbm{1} \left[\mathbb{E}\left[\ell \mid w^*, o, \theta\right]=\min _{w^{\prime}} \mathbb{E}\left[\ell \mid w^{\prime}, o, \theta\right]\right] P(\theta \mid \mathcal{H}) d \theta
\end{multline}
\end{Def}

\begin{assum}[Markov Interference Channel]
	We assume the interference state generating process can be represented as a stochastic matrix $P: \ncalS \times \ncalS \mapsto [0,1]$ composed of transition probabilities associated with specific state transitions. For this argument, we assume the radar's choice of waveform does not impact the state transition behavior, however the model can be generalized such that $P(s,a,s'): \ncalS \times \ncalA \times \ncalS \mapsto [0,1]$, which is the usual model for a Markov Decision Process (MDP). The model can also be further generalized to incorporate any $K^{\text{th}}$ order Markov process, as in \cite{thornton2022universal}.
	\label{assum:Markov}
\end{assum}

\begin{Def}[Collisions]
	Let $C_{t}$ correspond to a `collision' event, in which $w_{t}$ and $s_{t}$ contain energy in overlapping frequency bands during PRI $t$. More precisely, 
	\begin{equation}
		C_{t} = 
		\begin{cases}
			1, & \text{if } \sum_{i=1}^{d}\mathbbm{1}\{w_{t}[i] = s_{t}[i] \} > 0 \\
			0, & \text{if } \sum_{i=1}^{d}\mathbbm{1}\{w_{t}[i] = s_{t}[i] \} = 0.
		\end{cases}
	\end{equation}
	Further, we define $N_{c,t} = \sum_{i=1}^{d}\mathbbm{1}\{w_{t}[i] = s_{t}[i] \}$ to be the number of colliding sub-bands at PRI $t$.
\end{Def}

\begin{Def}[Missed Opportunities]
	Let $M_{t}$ correspond to a `missed opportunity' event, which occurs whenever the radar does not make use of available sub-channels. More precisely, this event is defined by
	\begin{equation}
		M_{t} = 
		\begin{cases}
			1, & \text{if } |w^{*}| - |w_{t}| > 0 \\
			0, & \text{otherwise} 
		\end{cases}
	\end{equation}
	where $w^{*}$ is the widest available grouping of sub-channels in the current interference state vector $s_{t}$, and $w_{t}$ is a binary vector which identifies the frequency bands utilized by the waveform actually transmitted by the radar at PRI $t$. Further, we define $N^{mo}_{t} = \max{\{|w^{*}| - |w_{t}|,0\}}$ to be the number of available sub-channels missed at PRI $t$.
\end{Def}

We now define the loss function for the learning algorithm $\ell_{t}: \ncalS \times \ncalW \mapsto [0,1]$, which is specified by
\begin{equation}
	\ell_{t} = \left\{
	\begin{array}{lr}
		1, & \text{if } N^{c}_{t} > 0\\
		\eta N^{mo}_{t}, & \text{if } N^{c}_{t} = 0
	\end{array}
	\right\}, 
	\label{eq:loss}
\end{equation}
where $\eta \in [0,1/N]$ is a parameter which controls the radar's preference for bandwidth utilization compared to SIR. This loss function is justified as follows. Collisions are likely to cause a missed detection, numerous false alarms, or performance degradation to neighboring systems. Thus, collisions are to be avoided at all costs. Given that collisions can be effectively avoided, the radar wishes to utilize as much available bandwidth as possible to improve resolution performance and ensure efficient utilization of the available spectrum. 

We note that the loss function expressed in (\ref{eq:loss}) is not necessarily unique. However, this formulation balances the trade-off between the two fundamental errors for this application in a manner that expresses the importance of each. In the numerical results, we note that the value of $\eta$ primarily impacts the radar's behavior at inflection points. For example, when it is unclear whether the radar should attempt to utilize a particular band or not, the value of $\eta$ dictates the radar's level of `risk aversion', effectively speaking. In Section \ref{se:numer}, we examine the impact of $\eta$ on radar behavior more precisely.

\section{Analysis}
In this section, we utilize the problem formulation developed in Section \ref{se:syst} to determine under what conditions, and in what sense, the TS policy provides performance benefits over the SAA policy, as well as the converse.
\label{se:analysis}
\begin{remark}
	Under Assumption \ref{assum:Markov} and a fixed time horizon $n$, the cumulative number of collisions and mixed opportunities experienced by the radar under any policy are random variables.
\end{remark}

Since the loss function specified in (\ref{eq:loss}) is a linear combination of missed opportunities and collisions, we may analyze the performance of any decision-making strategy by introducing a \emph{partial ordering} on the random variables which describe these performance statistics. In the development below introduces several of such orderings are introduced.

From the definition of missed opportunities and collisions, we note that for a fixed decision rule such as SAA, both events can be associated with specific state transition events. Given that the SAA policy is a fixed decision rule, the frequency of these events can be identified by the following decomposition of $P$:

\begin{prop}[Structure of Transition Matrix]
	Define three sets $\ncalA$, $\ncalB$, and $\ncalC$. Set $\ncalA$ consists of transitions $t_{ij}$ such that $|\pi_{\text{SAA}}(s_{i})| < |\psaa(s_{j})|$ and such that $|\psaa(s_{i})| = |\psaa(s_{j})|$ but $\psaa(s_{i}) \neq \psaa(s_j)$. In other words, the set of transitions such that acting on the SAA policy will produce a collision. Set $\ncalB$ consists of the transitions $t_{ij}$ such that $|\psaa(s_{i})| > |\psaa(s_{j})|$. Finally, set $\ncalC$ consists of $t_{ij}$ such that $\pi_{\text{SAA}}(s_{i}) = \psaa(s_{j})$.
\end{prop}

\begin{remark}[Effectiveness of SAA]
	If $s_{t} = s_{t-1}$ and $o_{t} = o_{t-1} = s_{t}$, then transmitting $w_{t} = \pi_{\text{SAA}}(o_{t-1})$ during each PRI results in zero collisions and zero missed opportunities. Equivalently, the SAA strategy is effective only when the diagonal elements of $P$, $p_{ii}$, are close to one, and when the probability of being in set $\ncalC$ is large.
\end{remark}

\begin{remark}[Effectiveness of Learned Policies]
	Learned policies will be effective under the following conditions: 1. long time horizon $n \rightarrow \infty$, 2. many elements of $P$ close to $0$ or $1$ (nearly deterministic), and when the probability of being in set $\ncalA$ or $\ncalB$ is large.
\end{remark}

Having generally identified under which conditions the SAA and TS policies will be effective, we now move towards a more specific characterization, employing tools from dynamic programming and statistical decision theory.

\begin{Def}[Optimal Long-Term Average Cost]
Given an initial state $s_{0}$, the optimal long-term average cost achievable by a policy is given by
	\begin{align}
		&\lambda^*\left(s_{0}\right) = \nonumber\\
		&\inf _{\pi \in \Pi} \limsup _{n \rightarrow \infty} \mathrm{E}_\nu\left[\frac{1}{n} \sum_{t=1}^n \ell\left(s_t, w_t \right) \mid s_{1}^{t-1}, w_{1}^{t-1}\right],
	\end{align}
where the infimum is taken over the set of all admissible decision rules $\Pi$.
\end{Def}

A policy which achieves the optimal long term average cost is known as the Bellman optimal policy, which is not necessarily unique.

\begin{Def}[Bellman Optimal Policy]
	\begin{align}
		\label{eq:Bellman}
		\pi^{*} = \min _{w_B} \sum_{s_{t+1}} P &\left(s_{t+1} \mid s_{1}^{t}\right) \\ \nonumber
		& \times\left[\ell\left(s_t+1, w_B \right)+\alpha \ell\left(s_1^{t}, w_1^{t}\right)\right],
	\end{align}
	where $\alpha \in [0,1]$ is a weighting factor called the \emph{discount rate}.
\end{Def}

\begin{figure*}
	\centering
	\includegraphics[scale=0.5]{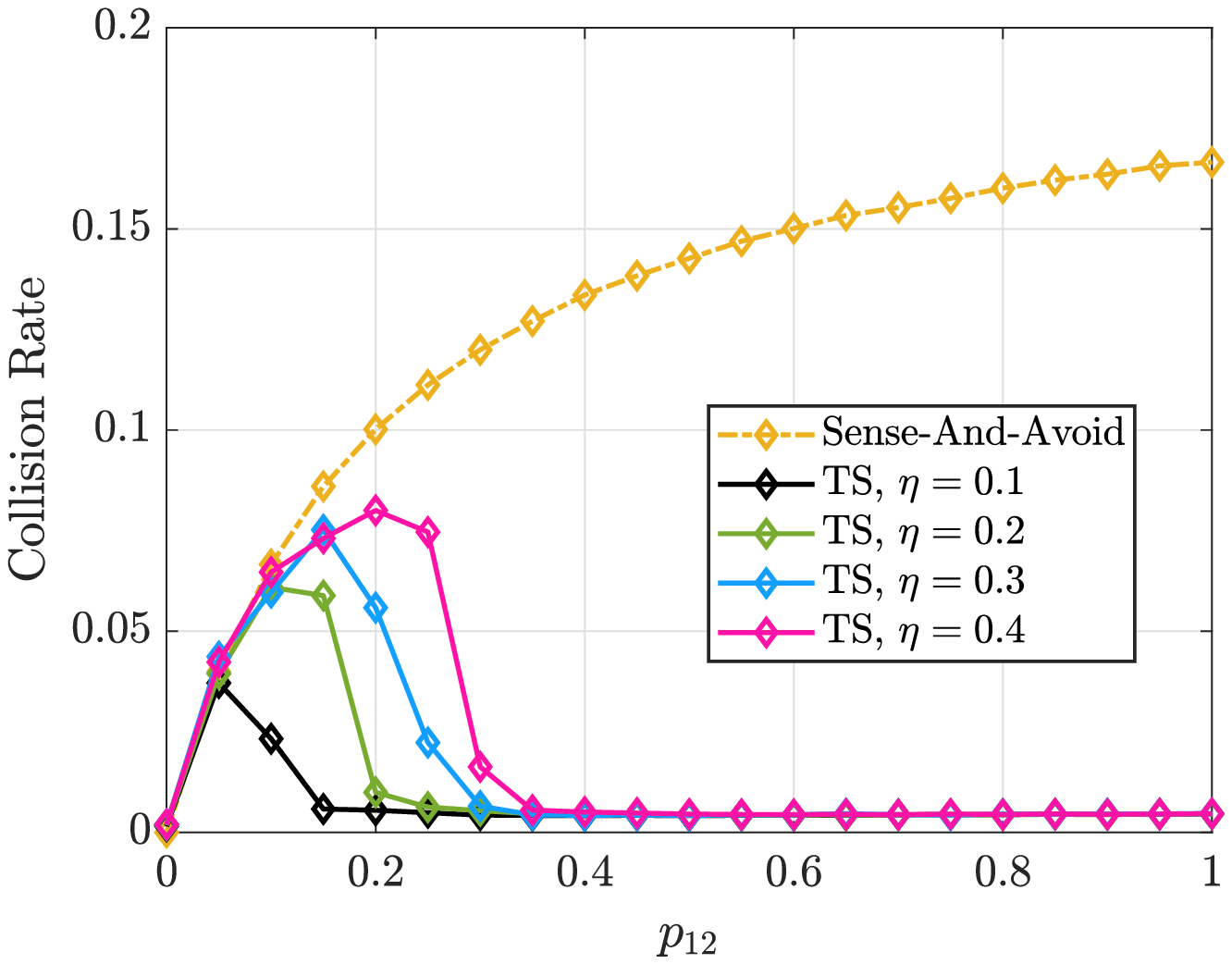}
	\includegraphics[scale=0.5]{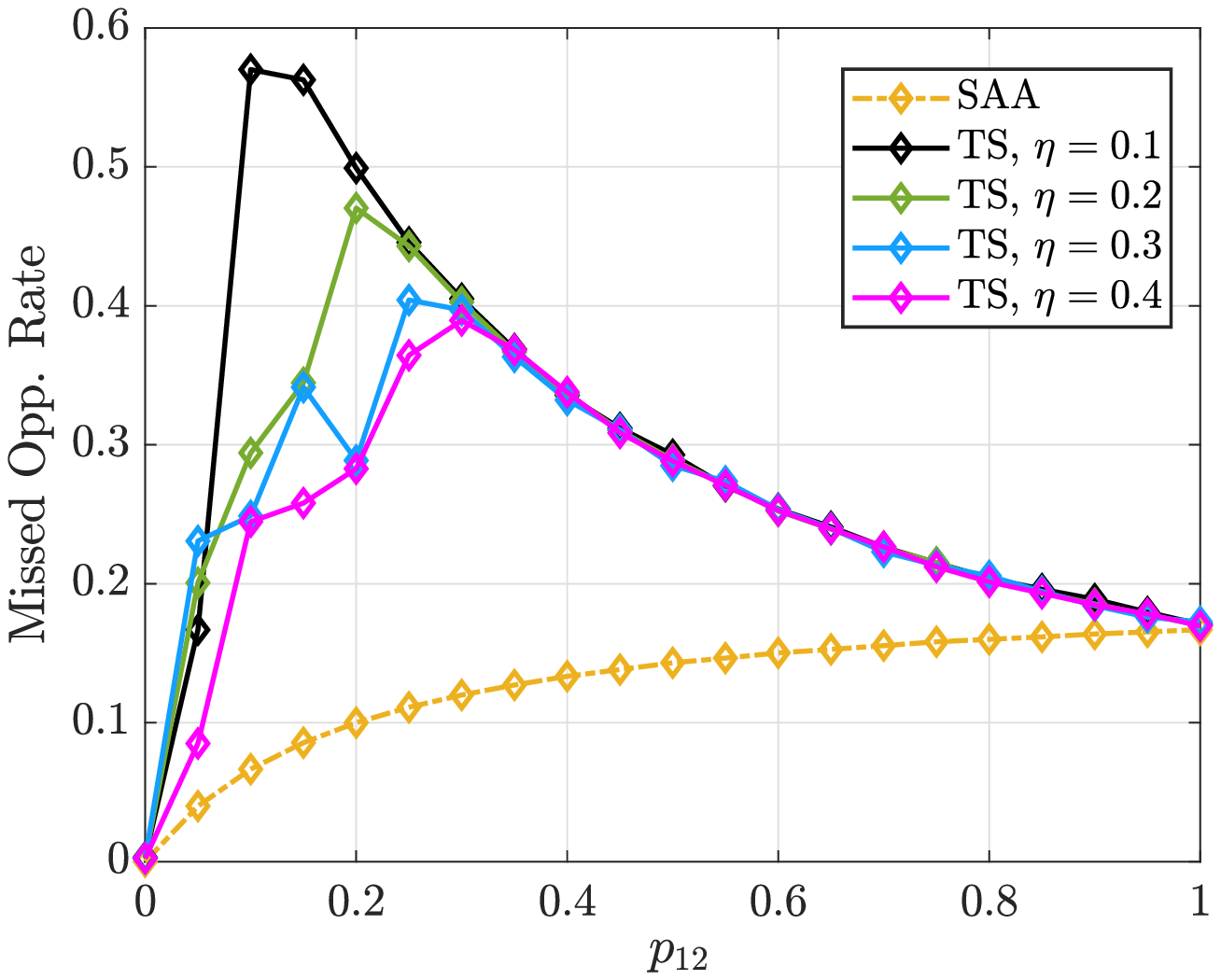}
	\caption{Analysis of the collision and missed opportunity rates as transition probability $p_{12}$ is varied. We observe that as the interference becomes increasingly dynamic ($p_{12} \rightarrow 1$), the learning-based approaches definitively outperform SAA over a time horizon of $n = 1e4$.}
	\label{fig:firstCase}
\end{figure*}

\begin{figure*}
	\centering
	\includegraphics[scale=0.5]{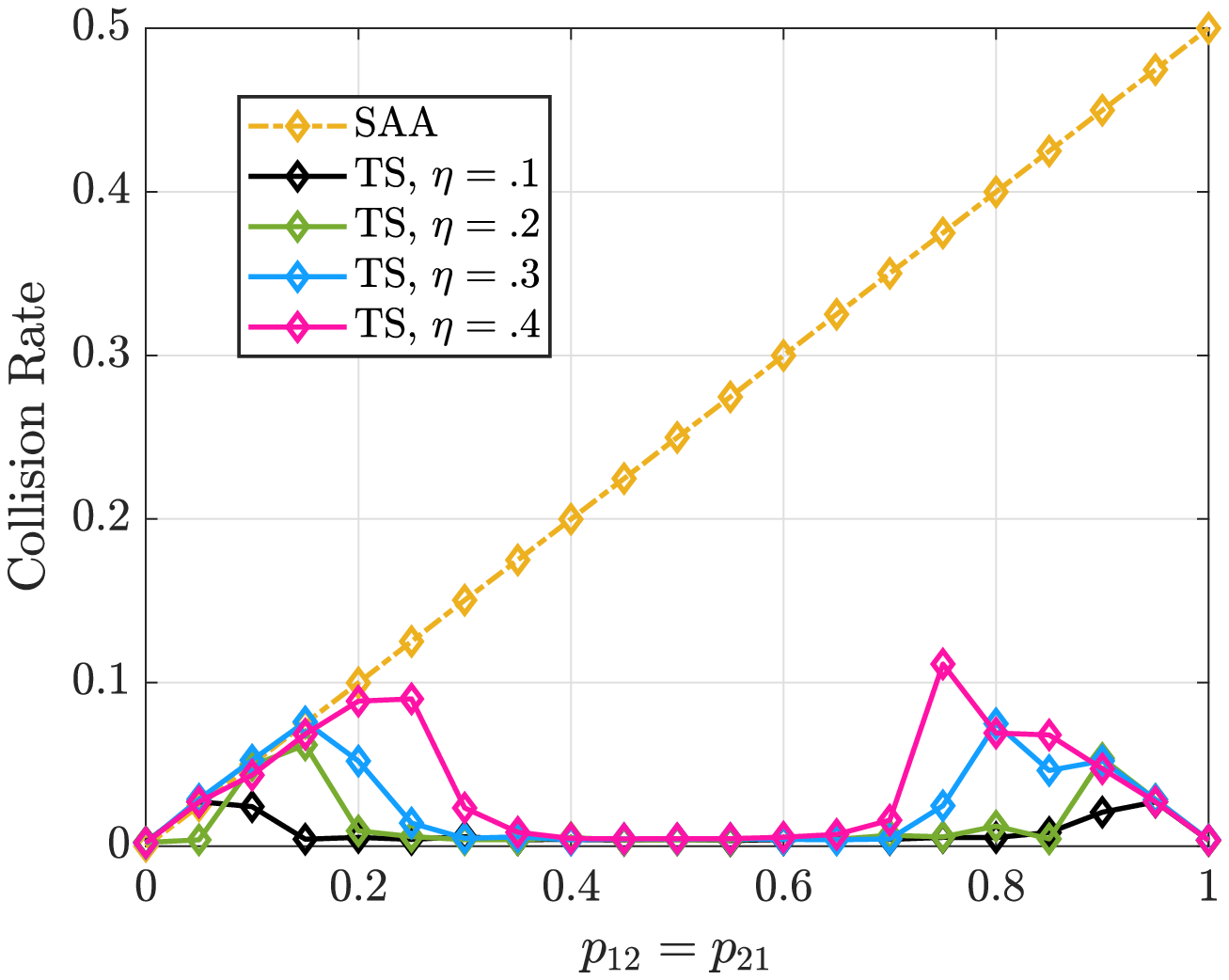}
	\includegraphics*[scale=0.5]{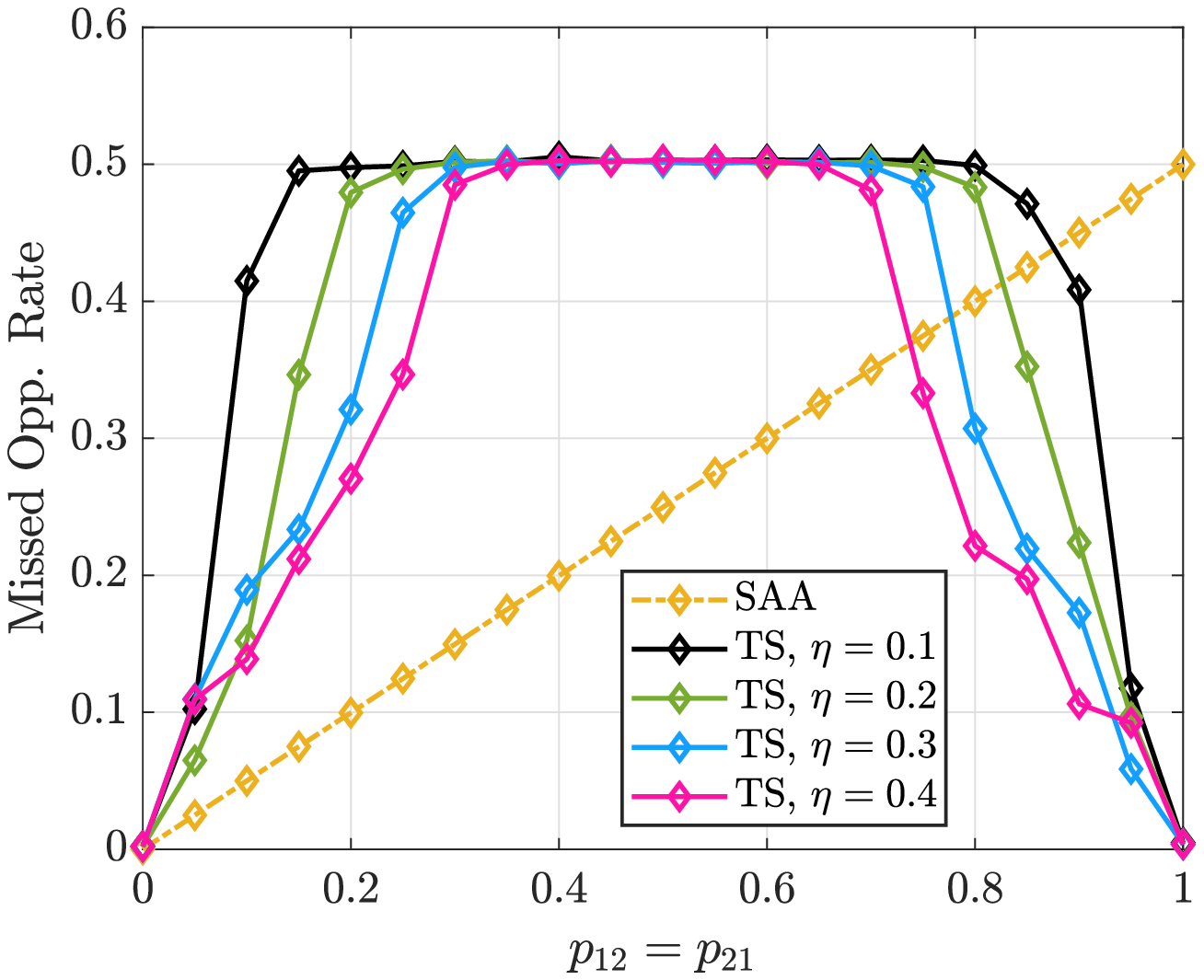}
	\caption{Collision and missed opportunity rates as transition probabilities $p_{12}$ and $p_{21}$ are jointly varied. As the interference becomes increasingly dynamic and determinsitic ($p_{12} = p_{21} \rightarrow 1$), the learning-based approaches definitively exceed SAA in performance.}
	\label{fig:joint}
\end{figure*}

The strongest measure of dominance that one waveform selection policy may have over another is statewise dominance, defined as follows:

\begin{Def}[Statewise Dominance]
	A waveform selection policy $\pi_{1}$ is said to dominate policy $\pi_{2}$ \emph{statewise} if $\ell(\pi_{1}(s)) \leq \ell(\pi_{2}(s))$ for every $s \in \ncalS$.
\end{Def}

\begin{remark}
	Utility maximizing strategies, such as the Bellman optimal policy specified in (\ref{eq:Bellman}), often sacrifice performance in states that occur with low probability, in order to improve performance in states that occur with high probability. Thus, we argue that while statewise dominance is a strong condition, it is not always the most meaningful notion of dominance for cognitive radar problems.  
\end{remark}

\begin{Def}[Long-Term Average Dominance]
	Let $\lambda_{\pi_{1}}$ be the long-term average cost associated with acting according to policy $\pi_{1}$ in some fixed environment $E$, specified by Markov transition kernel $P$. Policy $\pi_{1}$ is said to dominate $\pi_{2}$ on average if $\lambda_{\pi_{1}} < \lambda_{\pi_{2}}$.
\end{Def} 

\begin{Def}[First-Order Dominance]
	Policy $\pi_{1}$ is said to first-order dominate policy $\pi_{2}$ if $P(\ell(\pi_{1}) \geq x) \leq P(\ell(\pi_{2}) \geq x)$ for all $x \in \mathbb{R}$ and $P(\ell(\pi_{1}) \geq x) < P(\ell(\pi_{2}) \geq x)$ for some $x$.
\end{Def}

\begin{Def}[Second-Order Dominance]
	Let $F_{1}$ be the CDF of a performance statistic, namely here the observed loss as specified in (\ref{eq:loss}), under policy $\pi_{1}$ and likewise $F_{2}$ be the cdf under $\pi_{2}$. Then $\pi_{1}$ is said to second-order dominate $\pi_{2}$ if 
	\begin{equation*}
		\int_{-\infty}^x\left[F_1(t)-F_2(t)\right] d t \geq 0 
	\end{equation*}
	for all $x \in \nbbR$ with strict inequality at some $x$.
\end{Def}

\begin{remark}
	Second-order stochastic dominance is a necessary condition for first-order stochastic dominance. Thus, first-order dominance is a stronger condition.
\end{remark}

\begin{theorem}
	For any time-varying Markov interference channel, namely whenever some non-diagonal elements of $P$ are less than one, $\pi_{\text{TS}}$ weakly dominates dominates $\psaa$ in the long-term average sense. Further, as $t \rightarrow \infty$, $\pts$ dominates $\psaa$ in the second-order stochastic sense.
\end{theorem}

\begin{proof}
	See Appendix
\end{proof}

\begin{remark}
	For time-varying Markov interference channels that are quickly varying, namely, when the diagonal elements of $P$ are close to zero,
\end{remark}

\begin{figure}
	\centering
	\includegraphics[scale=0.5]{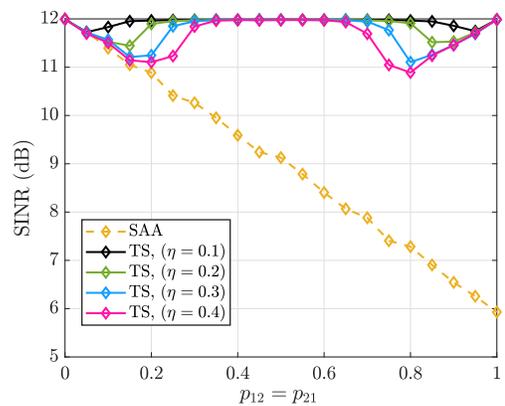}
	\caption{Average instantaneous SINR as transition probabilities $p_{21}$ and $p_{21}$ are jointly varied. As the channel becomes increasingly dynamic, the SINR under the SAA policy drops and reliability of target detection is diminished.}
	\label{fig:sinr}
\end{figure}

\begin{figure}
	\centering
	\includegraphics[scale=0.5]{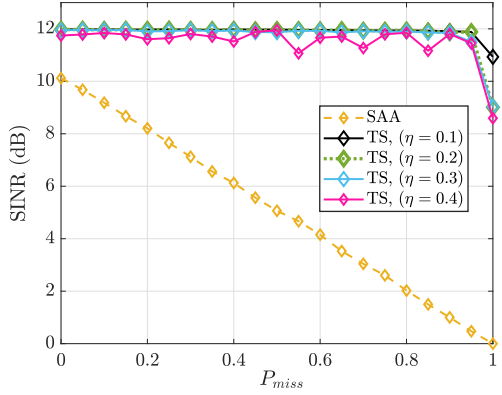}
	\caption{Impact of missed interference detections on average SINR. TS policies are generally robust to incorrect observations, while the SAA policy quickly degrades in performance with increasing $P_{miss}$.}
	\label{fig:miss}
\end{figure}

\begin{figure}
	\includegraphics[scale=0.5]{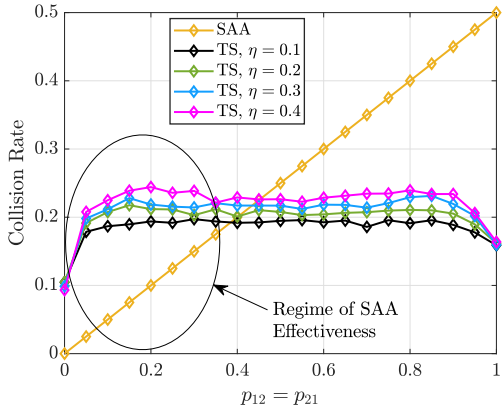}
	\caption{Collision rate of SAA and TS algorithms when time horizon is limited to $n=300$. We observe the regime where SAA is dominant.}
	\label{fig:lown}
\end{figure}

\section{Numerical Results}
\label{se:numer}
In these simulations, we examine the performance of the Thompson Sampling based waveform selection policy with an uninformative Gaussian prior. We consider a simple Markov interference channel with 2 states, having stationary transition probability matrix $P = [p_{11},p_{12};p_{21},p_{22}]$. The initial state $s_{0}$ is selected randomly.

In Figure \ref{fig:firstCase}, we observe the collision and missed opportunity rate over a time horizon of $n = 1e4$ as transition probability $p_{12}$ is varied. We observe that for low values of $p_{12}$ SAA and TS perform nearly equivalently in terms of collision rate. However, SAA provides an advantage in terms of reduced number of missed opportunities. We note that an inflection occurs around $p_{12} = 0.3$. As the channel becomes less stationary, the TS policy learns to avoid collisions effectively by avoiding certain actions altogether. Equivalently, the band is utilized more often, and the number of opportunities missed by the avoidant scheme of TS decreases. Thus, we observe that as the interference becomes more time-varying, TS performs much better than SAA, as expected.

In Figure \ref{fig:joint}, we observe the performance of each scheme as transition probabilities $p_{12}$ and $p_{21}$ are jointly varied. Values of $p_{12} = p_{21}$ near zero indicate a nearly stationary channel, while values $p_{12} = p_{21} \approx 1$ indicate rapidly fluctuating interference that changes deterministically. We note the performance can be segmented into three regions. For $p_{12} = p_{21} \leq 0.2$, the performance of TS closely mirrors SAA in terms of collisions, while incuring a slightly higher missed opportunity rate. In the regime of $0.2 < p_{12} = p_{21} < 0.6$, the TS approach effectively mitigates collisions at the expense of a missed opportunity rate of approximately $0.5$. In the final regime, $p_{12} = p_{21} > 0.6$ we note that the TS strategy outperforms SAA in terms of both missed opportunities and collisions. This is because the interference transitions become closer to deterministic as the transition probabilities approach $1$.

In Figure \ref{fig:sinr}, we see the received SINR\footnote{In the experiments, SINR is approximated using the well-known radar range equation, assuming an $INR$ of 14dB.} of these approaches for the case of jointly varied transition probabilities. We observe that for all values of $p_{12} = p_{21}$, the TS waveform selection strategy effectively maintains a high SINR, while SAA begins to perform poorly as the channel becomes increasingly time-varying. For rapidly varying channels, it can be concluded that SAA is untenable to maintain a sufficiently high probability of detection. Additionally, as the interference behavior becomes closer to deterministically varying, the TS approach is able to effectively predict transitions, and approaches missed opportunity and collision rates nearing zero.

Figure \ref{fig:miss} shows the performance degradation of SAA when the assumption of perfect interference observation is lifted. In this experiment, the radar fails to observe the interference state $s_{t}$ with probability $P_{\text{miss}}$ at each time slot, and instead observes a vector of all zeros. Once again, the time horizon is $n = 1e4$ and the transition behavior is held constant at $p_{12} = p_{21} = 0.3$. We observe that while TS is very robust to imperfect observations, the performance of SAA quickly degrades, as expected.

To highlight the regime where SAA is effective, we examine a drastically reduced time horizon, $n = 300$, in Figure \ref{fig:lown}. We observe that due to the learning time required by TS, the approach performs significantly worse than in the case of a longer time horizon, while the performance of SAA remains unchanged. However, we note that even in the case of a short time horizon, when $p_{12} = p_{21} > 0.4$, TS performs better than SAA in terms of collision rate.

From these results, we observe that generally, the learned policies are more robust to quickly-varying interference behavior and imperfect interference observations than SAA. On the other hand, SAA performs decidedly better than SAA when the time horizon is short, interference observation probability is high, and the channel is varying relatively slowly. Thus, it is important for the system designer considering the use of a learning-based algorithm to take into account the expected time horizon of the learning process as well as state observation probability. In general, fixed decision rules will break down when fundamental assumptions are violated, while learning-based approaches may be more robust. We note that for many radar applications, such as target tracking and electronic warfare, the state observation probability may be low, suggesting the value of learning may be very high for applications in which specific domain knowledge is not available \emph{a priori}.

\section{Conclusions, Insights, and Open Problems}
\label{se:concl}
In this contribution, we have sought to answer a basic question about the value of learning-based cognitive radar by examining the performance of a Thompson Sampling-based dynamic spectrum access strategy as compared to a rule-based strategy in a Markov interference channel. We showed that the performance of any policy, or decision function, can be analyzed by decomposing the set of channel transitions into three subsets which characterize collision, missed opportunity, and successful transition events. We argue that a similar procedure can be generalized to compute the expected performance of policies in $K^{\text{th}}$ order Markov channels for various applications, as in \cite{thornton2022universal}. For stationary channel dynamics, the performance of any ML strategy can be decomposed into the performance during the exploration period and performance during exploitation.

Our numerical results demonstrate that given a time horizon of $10^{4}$ pulse repetition intervals, Thompson Sampling performs at least as well as the sense-and-avoid policy in terms of average SINR. However, for \emph{nearly stationary} interference channels over a more limited time horizon, sense-and-avoid outperforms TS by utilizing additional bandwidth while maintaining a nearly equivalent SINR. For Markov interference channels which are very dynamic, namely which have diagonal transition elements close to zero, TS provides large performance benefits, and generally dominates SAA in the first-order stochastic sense. Particularly, for channels which are very dynamic, but are close to deterministic, TS provides the largest performance benefits compared to SAA.

Online learning theory tells us that the most important factors which influence the regret of algorithms are the number of actions, required dimensionality of the context vector, and the time horizon itself. Thus, systems designers planning to employ learning-based strategies should account for the allowable time in which learning may occur, limit the waveform catalog to a manageable size, and represent the environment as efficiently as possible.

We have also seen that the observability of the scene's state plays a large role in the value of any determinsitic waveform selection strategy. Machine learning algorithms are expected to provide large benefits over rule-based strategies if the observations are highly variable, which is the case when target measurements are incorporated into the model.

Future work will attempt to generalize this methodology to provide insights as to which environmental conditions make cognitive radar beneficial. For example, a more nuanced comparison involving more sophisticated fixed-rule based strategies would give greater insight into the practical value of cognitive radars. However, it is expected that next-generation cognitive radar systems will use a mixture of sequential decision making and rule-based strategies depending on the expected characteristics of the scene. Future work could examine using the definitions of stochastic dominance examined here to develop a principled algorithm selection scheme for a meta-cognitive radar.

\appendix

\begin{proof}[Proof of Theorem 1]
	Assuming perfect observability of the interference state (namely $o_{t} = s_{t} \quad \forall t$), the regret of $\psaa$ can be decomposed into two terms, the number of interference state transitions experienced, given by $R \in \nbbN$, and the sub-optimality gap $\Delta$ associated with each transition. Without loss of generality, let $R > 0$, $\Delta > 0$ be constants. Then it follows that SAA incurs regret linearly in $T$. The sub-linearity of TS regret has been well-established given a fixed number of arms $K$ and context dimension $d$, and is order $\tilde{\mathcal{O}}(d^{3/2}\sqrt{n})$. Thus, for any non-stationary Markov channel (namely, when some diagonal elements of $P$ are less than $1$), TS weakly dominates SAA in the asymptote. The extension to cases of non-perfect interference observability directly follow. However, in the finite horizon regime, precise specification of $K$, $T$, and $P$ are necessary to show stochastic dominance. It should be noted that for slowly varying channels, large values of $K$ and small $T$, SAA is more appealing.
\end{proof}

\bibliographystyle{IEEEtran}
\bibliography{smithBib.bib}{}

\end{document}